\documentclass[10pt,conference,twocolumn]{IEEEtran}

\usepackage{cite}
\usepackage{balance}
\usepackage{color}
\usepackage{epsfig}
\usepackage{epstopdf}
\usepackage{graphicx}
\usepackage{tabularx}
\usepackage{varwidth}

\usepackage{amssymb,amsmath,amsthm}
\newtheorem{lemma}{Lemma}

\usepackage{bm}
\usepackage{algorithm}%

\newfloat{routine}{htbp}{loa}
\floatname{routine}{Routine}
\makeatletter\newcommand\l@subroutine{\@dottedtocline{1}{1.5em}{2.3em}}\makeatother

\usepackage{algpseudocode}
\usepackage{pict2e}
\makeatletter
\def\BState{\State\hskip-\ALG@thistlm}

\makeatother

\newcommand{\mU}{\mathcal{U}}
\newcommand{\mR}{\mathcal{R}}
\newcommand{\mS}{\mathcal{S}}
\newcommand{\mC}{\mathcal{C}}
\newcommand{\mP}{\mathcal{P}}
\newcommand{\mV}{\mathcal{V}}

\newcommand{\mN}{\mathcal{N}}
\newcommand{\mK}{\mathcal{K}}
\newcommand{\m}[1]{\mathcal{#1}}

\usepackage{colortbl}

\newcounter{remarkCounter}

\newcounter{probCounter}

\begin{document}
\title{Octopus: A Cooperative Hierarchical Caching Strategy for Cloud Radio Access Networks}


\author{Tuyen~X.~Tran~\IEEEmembership{Student Member,~IEEE,}and
Dario~Pompili~\IEEEmembership{Senior Member,~IEEE,}\\\IEEEauthorblockA{Department of Electrical and Computer Engineering\\
Rutgers University--New Brunswick, NJ, USA\\
E-mails: \{tuyen.tran, pompili\}@cac.rutgers.edu}
}

\maketitle

\thispagestyle{empty}


\begin{abstract}


Recently, implementing Radio Access Network (RAN) functionalities on cloud-based computing platform has become an emerging solution that leverages the many advantages of cloud infrastructure, such as shared computing resources and storage capacity, while lowering the operational cost. In this paper, we propose a novel caching framework aimed at fully exploiting the potential of such Cloud-based RAN (C-RAN) systems through cooperative hierarchical caching which minimizes the network costs of content delivery and improves users' Quality of Experience (QoE). In particular, we consider the cloud-cache in the cloud processing unit (CPU) as a new layer in the RAN cache hierarchy, bridging the capacity-performance gap between the traditional edge-based and core-based caching schemes. A delay cost model is introduced to characterize and formulate the cache placement optimization problem, which is shown to be NP-complete. As such, a low complexity, heuristic cache management strategy is proposed, constituting of a proactive cache distribution algorithm and a reactive cache replacement algorithm. Extensive numerical simulations are carried out using both real-world YouTube video requests and synthetic content requests. It is demonstrated that our proposed Octopus caching strategy significantly outperforms the traditional caching strategies in terms of cache hit ratio, average content access delay and backhaul traffic load.

\end{abstract}

\begin{IEEEkeywords}
Hierarchical caching; cooperative caching; cloud Radio Access Networks; C-RAN. 
\end{IEEEkeywords}

\section{Introduction}
\textbf{Overview:} Over the last few years, the raise of social networks (Facebook, Twitter, Instagram...), entertainment applications and multimedia content providers (YouTube, Netflix, etc.) has generated a burgeoning traffic demand on wireless mobile network. It is expected that mobile traffic will increase by 10-fold by 2018 while Content Delivery Networks (CDNs) account for $36\%$ of the Internet traffic \cite{cisco2012cisco}. This demand has fundamentally shifted from being the steady increase in traffic for connection-centric communications, such as phone calls and text messages, to the explosion of content-centric communications, such as video streaming and content sharing. The mobile cellular network architectures of today are, however, still designed with a connection-centric communication mindset. Moreover, the myriad technological advances proposed for beyond 4G and 5G mobile networks are mostly geared towards capacity increase, which is fundamentally constrained by the limited radio spectrum resources as well as the diminishing investment efficiency for operators. In order to support massive content delivery in an affordable way, improving network capacity alone is not sufficient and has to be accompanied with innovations at higher layers.

In today's mobile networks, computing and storage capabilities are already ubiquitous, both at the Base Stations (BSs) and on user devices themselves. In-network caching, which proactively stores popular contents at the network nodes (preferably close to the users), has therefore become a very promising solution to reduce latency and network costs of content delivery. When a user requests for a content that is cached in a nearby network node, that network node can directly provide the content to the user, rather than downloading it from the origin server in the CDN. 

A caching system in RAN could also leverage the current trends towards a Cloud-based Radio Access Network (C-RAN)~\cite{whitepaper13}. In C-RAN, the computational functionalities of the BSs are implemented in a common cloud processing unit (CPU) which can be hosted in a small data center. The centralized nature of C-RAN enables highly dynamic management of on-demand resource allocation~\cite{pompili2015dynamic, pompili15commag} and collaborative communications~\cite{tran2015mass, tran2016secon}. Additionally, the CPU with strong computing resources and storage capacity can provide a central port for traffic offload and content management to handle growing Internet traffic from mobile users. This directly translates into Capital Expenditure (CapEx) and Operational Expenditure (OpEx) reduction as well as user experience improvement \cite{sarkissian2012business}.

In this paper, we leverage C-RAN architecture and propose a novel cooperative hierarchical caching strategy, Octopus, with distributed \emph{edge-caches} at the BSs and the intermediate \emph{cloud-cache} at the CPU. The deployment of edge-caches and cloud-cache are complementary and interoperable. The overall system design and optimization solution for such caching systems, which could involve researches at multiple layers, is a complex problem and would go well beyond the scope of this paper. Here, assuming that the information about content popularity is available, we focus on the cache placement optimization problem which addresses the questions of \emph{what} and \emph{where} to place the contents among the cache nodes. In particular, we identify an efficient cooperative in-network caching strategy aiming at minimizing the total expected cost of content delivery.

\begin{figure}
 \centering
\includegraphics[scale = 0.65]{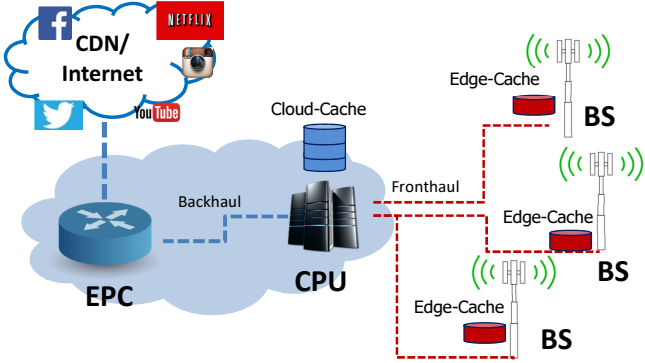} 
\caption{Illustration of a C-RAN caching system where the cloud-cache is deployed at the CPU and the edge-caches are deployed at the BSs.}\label{fig:cran_nw}
\end{figure}



\textbf{Related Works:} Recently, some prior works on content caching in cellular networks have proposed to deploy caches at the edges of the RAN (i.e., the BSs) \cite{bastug2014living, ahlehagh2014video, golrezaei2012femtocaching}. In \cite{bastug2014living}, the authors propose to alleviate backhaul congestion via proactive caching at the small cell BSs, whereby files are proactively cached during off-peak hours based on file popularity and correlations among users and file patterns. In \cite{golrezaei2012femtocaching}, the notion of femtocaching is introduced, in which the femtocell-like BSs are used to form a distributed caching network that assists the macro BS to handle requests of popular files that have been cached. The authors also consider the case that each mobile user can access multiple caches, and present approximation algorithms  to address the distributed cache assignment problem. In \cite{Gharaibeh2015online}, the authors also consider cooperative caching, where each user can accesss multiple caches from neighboring BSs, and propose an online caching algorithm that does not require prior knowledge about the content popularities in order to address the problem of minimizing the total cost paid by the content providers. The work in \cite{ahlehagh2014video} utilizes User Preference Profiles (UPPs) of active users in a cell to derive RAN-aware reactive and proactive video caching policies.

While offering great potential to bring popular content closer to the users, the aforementioned caching schemes only rely on the deployment of edge-caches. Hence, due to limited cache size at the BSs (compared to the very large amount of popular content), these \emph{edge-only} caching schemes suffer from high cache miss ratio. To compensate for the relative small cache size at the BSs, the authors in~\cite{wang2014cache} consider caches both in the RAN edge and in the Evolved Packet Core~(EPC). Along this line, the techniques in \cite{ahlehagh2014video} are further extended to a hierarchical caching scheme in \cite{ahlehagh2012hierarchical} where the gateways in the EPC also have video caches. 
While it is possible to implement relatively large cache size at the EPC to improve the cache hit ratio, fetching content from EPC to the BSs still undertake considerable delay due to the involvement of multiple intermediate network components.

{\bf{Our Contributions:}}
Unlike existing approaches, we consider the deployment of edge-caches in cooperation with the additional cloud-cache at the CPU in a C-RAN. Such cloud-cache presents a new layer in the RAN cache hierarchy, bridging the gap between the edge-based caching (small cache size, low access latency) and core-based caching schemes (large cache size, high access latency).


We formulate a cache management optimization problem aiming at minimizing the total average delay cost, subject to the cache size constraint at each node. We show that this is an NP-hard problem and propose a low-complexity, heuristic 
strategy involving a proactive cache distribution (PCD) algorithm and a reactive cache replacement algorithm (RCR). In particular, the PCD algorithm starts with an empty cache set, and incrementally places content files in the caches until they are all full. The PCD algorithm yields a solution with at least $\frac{1}{2}$ of the optimal value. Such solution is further improved by running the RCR algorithm which determines whether to replace a cached file with a new one when it is downloaded to the local RAN due to a cache miss. 

We carry out extensive numerical simulations using both the real world YouTube video request trace and the Zipf-based popularity model. It is demonstrated that Octopus significantly outperforms traditional caching deployment architectures and cache management algorithms in terms of cache hit ratio, average content access latency and backhaul traffic load. The reduction in backhaul traffic load and content access latency can be directly translated into sizable cost savings in both backhaul and transport OpEx as well as user experience improvement.

{\bf{Paper Organization:}}
The remainder of this paper is organized as follows: in Sect.~\ref{sec: model}, we present the system description and formulate the cache management optimization problem; in Sect.~\ref{sec:efficient_alg}, we propose the efficient algorithm for the cache management problem; performance evaluation via numerical simulations is presented in Sect.~\ref{sec:perform_eval} and, finally, Sect.~\ref{sec:conclusion} concludes the paper. 

\section{Caching System Model} \label{sec: model}
In this section, we describe the cooperative hierarchical caching system. The delay cost model is then introduced in order to formulate the cache management optimization problem.

\subsection{System Description}

Let us consider a C-RAN that consists of a set $\mR = \left\{ {1,2,...R} \right\}$ of $R$ BSs distributed in $R$ corresponding cells and a set $\mU = \left\{ {1,2,...U} \right\}$ of $U$ \emph{active} users. All the BSs are connected to a  common CPU via low-latency, high-bandwidth fronthaul links as illustrated in Fig.~\ref{fig:cran_nw}. The collection of files available for download is denoted as $\m{F} = \left\{ {f_1,f_2,...f_F} \right\}$, in which the size $s_i~[\rm{MB}]$ of each file is assumed to be the same, as considered in \cite{golrezaei2012femtocaching}. This assumption is mainly used for notational convenience, and could be easily lifted by considering a finer packetization, and breaking longer files into blocks of the same length.
We define the popularity distribution of the files as $\m{P} = \left\{ {{p_1},{p_2},...{p_F}} \right\}$, i.e., the probability of the $k$th file being requested from any user in the network is $p_k$. While predicting the file popularity is a challenging task in terms of accuracy and scalability, the recent advances in machine learning and data mining techniques have made significant progress on achieving this goal. Such techniques could involve analysing data from popular web sites, news papers and social networks to determine, around a specific BS, what kinds of contents people like, search for, and what are the consumer  profiles of these people.

We consider that each user $u$ only connects to and receives data from the nearest BS (in terms of signal strength), which we later refer to as the user $u$'s home BS, and denote $\mU_r \subseteq \mU$ as the set of users served by BS $r$. Further extension to the system employing Coordinated Multi-Point (CoMP) transmissions, where each user can be served by multiple BSs, will be explored in a separate work. We consider that each BS $r$ is equipped with an \emph{edge-cache}, denoted as $\mC_r$, with storage capacity of $M_r~[\rm{files}]$, and the CPU is equipped with a \emph{cloud-cache}, denoted as $\mC_0$, with storage capacity of $M_0~[\rm{files}]$ (usually $M_0 \gg M_r, r = 1,...R$). To describe the cache placement decision, i.e., which files stored in which caches, we define the cache placement ground set as, 
\begin{equation} \label{eq:ground_set}
\mV = \left\{ {{f_{10}},{f_{20}},...{f_{F0}},...{f_{1R}},{f_{2R}},...{f_{FR}}} \right\},
\end{equation}
where $f_{ir}$ denotes the copy of file $f_i$ in cache $\mC_r$. Note that the indexing of caches $\mC_r$'s, $r = 0,1,...R$, includes all the edge-caches and cloud-cache. In the subsequent analysis, unless otherwise stated, we will refer to the file $f_i$ and $f_{ir}$ interchangeably. The ground set $\mV$ can be partitioned into $R+1$ disjoint sets, $\mV_0,\mV_1,...\mV_R$, where ${\mV_r} = \left\{ {{f_{1r}},{f_{2r}},...{f_{Fr}}} \right\}$ is the set of all files that might be placed in the cache $\mC_r$. Hence, we can write $\mC_r \subseteq \mV_r$. A feasible cache placement decision, denoted as ${{\mC}} = \left\{ {{\mC_0},{\mC_1},...{\mC_R}} \right\}$, must satisfy the storage capacity constraints as follows,
\begin{equation}
\left| {{\mC_r}} \right| \le {M_r},\forall r = 0,1,...R.
\end{equation}


In the current 4G cellular network, the eNodeBs are inter-connected via the X2 interface that is designed for exchanging control information or users' data buffer during handover \cite{sesia2009lte, robson2012small}. While this X2 interface is available for limited data transfer, it cannot be exploited for inter-cache data transfer and hence the eNodeBs cannot share their cache contents directly. In contrast, the BSs in cloud-based RAN are all connected to the common CPU via high-bandwidth, low-latency CPRI links for user data transportation \cite{morant2015optical, chanclou2013optical, bartelt2015fronthaul}. This allows each BS to retrieve cache contents from the neighboring BSs  via a ``U-turn'' (BS-CPU-BS) using fronthaul links. Note that retrieving cache data from neighboring BSs is more latency- and cost-effective than fetching content from the original remote server in the CDN via the backhaul network \cite{wang2014cache, Gharaibeh2015online, pantisano2014network}. In this paper, we consider a \emph{cooperative hierarchical caching} strategy where the cloud-cache and edge-caches collaboratively form an ``octopus-like'' caching network connected via fronthaul links. This scheme fully exploits the extra degrees of cooperation brought by C-RAN to pool the resources and increase cache hit ratio, reducing outbound requests to the higher level network elements.

In the proposed system, we consider that there is a Central Cache Manager (CCM) implemented at the CPU to monitor all the requests generated from users within the local RAN, and is responsible to make cache (re)placement decision. In addition, leveraging the powerful processing capability at the CPU, one can implement sophisticated learning and prediction algorithms to estimate the content popularity information $\m{P}$. While the actual content files are physically stored in the separated caches, a global indexing table can be maintained by the CCM to facilitate content lookup and cache management. When a user $u \in \mU_r$ makes a request for the file $f_i$ that is already stored in the local edge-cache $\mC_r$, it can directly download the file $f_i$ from $\mC_r$ without incurring traffic on the fronthaul and backhaul links. If the requested file $f_i$ is not stored in the local edge-cache, the request is forwarded to the CCM at the CPU. Upon receiving the request for file $f_i$ from BS $r$, the CCM will firstly search for $f_i$ in the cloud-cache $\mC_0$, and then in the neighboring caches of BS $r$, i.e., $\mC_k$'s, $k \neq r$. Once $f_i$ is found in one of the caches, the CCM will direct the user to download the file directly from that cache via fronthaul links; otherwise the user will download the file from the origin server in the CDN, incurring traffic in the backhaul links. In Fig.~\ref{fig:octopus} we illustrate the overview of our proposed Octopus caching system with three use cases within which the requested content is retrieved from the local edge-cache (user 1), the neighbor edge-cache (user 2) and the cloud-cache (user 3), respectively.

Although the cache storage can be relatively large, the number of files that can be stored in each cache is limited compared to the total number of files available. When a user request for a file that is not available in the cache of its home BS, it has to retrieve the file from other places, incurring additional access delay and fronthaul (and possibly backhaul) bandwidth consumption. While the cost associated with access delay and bandwidth consumption are often proportional and interchangeable, we will focus on the cost model for content access delay as considered in \cite{golrezaei2012femtocaching, wang2014cache}. The reduction of such cost directly translates into users' quality of experience (QoE) improvement. It is therefore imperative to design an efficient cache management strategy so as to minimize the expected network cost and content access delay. A cache management strategy involves the cache distribution decision, i.e., placing files in the caches, and the cache replacement decision, i.e., updating existing files in the caches. In the following, we first introduce the cost model to characterize the cache management scheme and then formulate the underlying optimization problem. 

\begin{figure}
 \centering
\includegraphics[scale = 0.6]{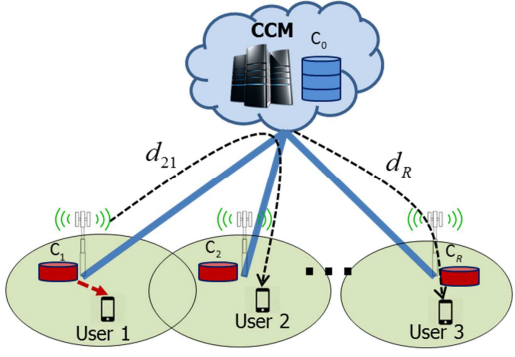} 
\caption{Illustration of Octopus caching system constituted of cloud-cache $\m{C}_0$ and edge-caches $\m{C}_1,...\m{C}_R$ which can share cached contents via fronthaul links. Example: requests from user 1 (in cell 1) and user 2 (in cell 2) are retrieved from edge-cache $\m{C}_1$, request from user 3 in cell $R$ is retrieved from cloud-cache.}\label{fig:octopus}
\end{figure}

Let $d_r$ denotes the delay cost of transferring a file from the cloud-cache to BS $r$ via fronthaul link, which we assume to be the same as the delay cost of retrieving that file from edge-cache $\mC_r$ to the CPU. Let $d_0$ denotes the delay cost incurred when a user downloads a file from the original server in the CDN. Furthermore, we assume the cost of transferring a file from cache of BS $k$ to BS $r$ is $d_{rk}$. In practice, $d_0$ is usually of many-fold higher than $d_r$ and $d_{rk}$ \cite{pantisano2014network, wang2014cache, gkatzikis2015clustered}. This makes it cost-effective to retrieve content from the in-network caches whenever possible rather than downloading them from the CDN. We consider that the incurred delay cost of a user downloading a file directly from its home BS's cache is zero. In order to formulate the cache management problem, let us define the following 0-1 decision variables for a given cache placement decision $\mC$ as follows, 
\begin{equation} \label{eq:c_ir}
{c_{ir}}(\mC) = \left\{ {\begin{array}{*{20}{l}}
1&{f_{ir} \in {\mC_r}},\\
0&{{\rm{otherwise,}}}
\end{array}} \right. 
\end{equation}
\begin{equation}
x_{ir}^k = \left\{ {\begin{array}{*{20}{l}}
1& \text{if request for file } i {\text{ from BS }}r \text{ is retrieved}\\ &{\text{ from cache }}{\mC_k}, k \neq r, k \in \left\{ {0,1,...R} \right\}, \\
0&{{\rm{otherwise,}}}
\end{array}} \right.
\end{equation}
\begin{align}
x_{ir}^{R+1} = \left\{ {\begin{array}{*{20}{l}}
1& \text{if request for file } i {\text{ from BS }}r \text{ is retrieved} \\ &{\text{from the CDN}}, \\
0&{{\rm{otherwise,}}}
\end{array}} \right. \\ \nonumber
\,\, \forall i \in \m{F}, \, r \in \mR.
\end{align}
Since each request should only be downloaded from the nearest possible place (having lowest cost), we impose the following constraint,
\begin{equation}
\sum\nolimits_{k = 0}^{R + 1} {x_{ir}^k}  = 1,\forall r = 1,...R.
\end{equation}
For a given cache placement decision ${{\mC}}$ and the popularity distribution $\m{P}$, we can calculate the average delay cost of user $u\in \mU_r$ as,
\begin{equation} \label{eq:delay_u}
{{\bar D}_u} = \sum\limits_{i \in \m{F}} {{p_i}(x_{ir}^0{d_r} + \sum\limits_{k \in \m{R}\backslash r} {x_{ir}^k{d_{rk}}}  + x_{ir}^{R + 1}{d_0}).} 
\end{equation}
Using this cost model, the cache management optimization problem is formulated in the next subsection.

\subsection{Problem Formulation}
Given the content popularity distribution $\mP$ and the constrained storage capacities of the caches, we wish to find an efficient cache management strategy in order to minimize the total average delay cost of all the users in the network. In particular, we consider a \emph{dynamic} cache management strategy that involves proactively distributing content files in the caches and reactively updating the cached files. Notice that multiple copies of the same file can be stored at different caches. The underlying optimization problem to realize the proposed strategy can be formulated as follows,
\begin{subequations} \label{eq:prob0}
\begin{align} \label{eq:prob0_a}
\mathop {\min }\limits_\mC  \hspace{0.5cm} &\sum\limits_{u \in \m{U}} {{\bar{D}_u}},\\
\label{eq:prob0_b}
\text{s.t.} \hspace{0.5cm} &\sum\limits_{i\in \m{F} } {{c_{ir}(\mC)}}  \le {M_r}, \hspace{0.3cm} \forall r = 0,1,...R, \\\label{eq:prob0_c}
&\sum\limits_{k = 0}^{R + 1} {x_{ir}^k}  = 1,\,\,\forall r = 1,...R,\\ \label{eq:prob0_d}
&x_{ir}^k \le c_{ik}(\mC), \,\,\forall r=1,...R, \, k = 0,1,...R.
\end{align}
\end{subequations}
with ${{\bar{D}_u}}$ given in (\ref{eq:delay_u}). The objective function in (\ref{eq:prob0_a}) represents the total average delay cost experienced by all the users in the network. The constraint in (\ref{eq:prob0_b}) imposes the cache storage capacities and the constraint in (\ref{eq:prob0_d}) ensures that a content file can be retrieved from a cache only if it has been stored in that cache. From constraint (\ref{eq:prob0_c}), we can substitute $x_{ir}^{R+1}$ by $1 - \sum\nolimits_{k = 0}^R {x_{ir}^k} $ into (\ref{eq:delay_u}), and problem (\ref{eq:prob0}) can be recast as a problem of maximizing the average delay cost reduction, expressed as,
\begin{subequations} \label{eq:prob}
\begin{align} \label{eq:prob_a}
\hspace{-0.2cm} \mathop {\max }\limits_\mC   \hspace{0.1cm} &\sum\limits_{u \in \m{U}} {\sum\limits_{i \in \m{F}} {{p_i}\Bigl( {x_{ir}^r{d_0} + x_{ir}^0\left( {{d_0} - {d_r}} \right) + \sum\limits_{k \in \m{R} \backslash r } {x_{ir}^k\left( {{d_0} - {d_{rk}}} \right)} } \Bigr)} },\\
\label{eq:prob_b}
\hspace{-0.2cm} \text{s.t.} \hspace{0.1cm} &\sum\limits_{i \in \m{F} } {{c_{ir}(\mC)}}  \le {M_r}, \hspace{0.3cm} \forall r = 0,1,...R,  \\ \label{eq:prob_c}
&\sum\nolimits_{k = 0}^R {x_{ir}^k}   \le 1,\,\,\forall r = 1,...R, \, i = 1,...F,\\ \label{eq:prob_d}
&x_{ir}^k \le c_{ir}(\mC), \,\,\forall r=1,...R, \, k = 0,...R, \, i = 1,...F.
\end{align}
\end{subequations}

The objective function in (\ref{eq:prob}) can be seen as the sum of \emph{utility value} seen by each user and our goal here is to maximize the sum utility value seen by all users. It can be shown that this problem is NP-complete (please refer to Appendix) and global optimal solution usually possesses exponential computational complexity which is impractical to implement. Therefore, our approach aims for a low-complexity, suboptimal solution that can be implemented in practical system. In particular, we will show that problem (\ref{eq:prob}) belongs to the classical class of problems of maximizing a \emph{monotone submodular function} over a \emph{matroid} constraint\cite{calinescu2011maximizing, nemhauser1978analysis}. We then propose a greedy cache management solution for problem (\ref{eq:prob}) consisting of a cache distribution algorithm and a backtracking cache replacement algorithm.

\section{Proposed Approach} \label{sec:efficient_alg}
We start this section by presenting some essential background material and intuition of our approach for the cache management problem. We then present the description of our proposed cache distribution and replacement algorithms.

\subsection{Preliminaries}

In the following, we provide the basic definitions of matroids and submodular functions \cite{wolsey2014integer}, which will be used in the analysis in the next subsection.

{\bf{Matroids.}} A \emph{matroid} is a pair $\left( {\mV,\m{I}} \right)$ such that $\mV$ is a finite set, and $\m{I} \subseteq 2^{\mV}$ is a collection of subsets of $\mV$ satisfying the following two properties
\begin{itemize}
\item $\m{I}$ is downward closed, i.e., if $A \subseteq B \subseteq \mV$ and $B \in \m{I}$ then $A \in \m{I}$.
\item If $A, B \in \m{I}$ and $\left| A \right| < \left| B \right|$, then there exists $ e \in A\backslash B$ such that $B \cup \left\{e \right\} \in \m{I}$.
\end{itemize}

Matroids generalize the concept of linear independence found in linear algebra to general sets, and sets in $\m{I}$ described above are called \emph{independence}. One of the important applications of matroids is the concept of matroid constraint defined via the \emph{partition matroid}. Consider a finite ground set $\mV$ that is partitioned into $n$ disjoint sets $\mV_1, \mV_2,...\mV_n$ with associated integers  $m_1,m_2,...m_k$, a partition matroid $\m{I}$ is given as,
\begin{equation} \label{eq:matroid_def}
\m{I} = \left\{ {A \subseteq \mV:\left| {A \cap {\mV_i}} \right| \le {m_i},\forall i = 1,...n} \right\}.
\end{equation}

{\bf{Subodular functions.}} Consider a finite ground set $\mV$, a set function $g:{2^\mV} \to \mathbb{R}$ is submodular if for all sets $A,B \subseteq \mV$,
\begin{equation}
g\left( A \right) + g\left( B \right) \ge g\left( {A \cup B} \right) + g\left( {A \cap B} \right).
\end{equation}
Given a submodular function $g:{2^\mV} \to \mathbb{R}$ and $A,S \subset \mV$, the function $g_A$ defined by ${g_A}\left( S \right) = g\left( {A \cup S} \right) - g\left( A \right)$ is also submodular, and if $g$ is monotone then $g_A$ is also monotone. For $i \in \m{V}$, we abbreviate $A \cup \{i \}$ by $A + i$. Let ${g_A}\left( i \right) = g\left( {A + i} \right) - g\left( A \right)$ denote the marginal value of an element $i\in \mV$ with respect to the subset $A \subseteq \mV$. Then, $g$ is submodular if for all $A \subseteq B \subseteq \mV$, and for all $i\in \mV \backslash B$ we have,
\begin{equation}
{g_A}\left( i \right) \ge {g_B}\left( i \right).
\end{equation}
Intuitively, submodular functions capture the concept of diminishing returns: as the set becomes larger the benefit of adding a new element to the set will decrease. The function $g$ is monotone if for $A \subseteq B \subseteq \mS$, we have $g\left( A \right) \le g\left( B \right)$. 


\subsection{Proposed Algorithms}
We exploit the special structure of problem (\ref{eq:prob}) to formulate it as the problem of maximizing a submodular function subject to matroid constraints. In particular, motivated by the approach in \cite{golrezaei2012femtocaching}, we will show that the constraints in (\ref{eq:prob}) can be expressed as the independent sets of a matroid and the objective function can be expressed as a monotone submodular function.

Firstly, recall that every cache placement decision $\mC$ is a subset of the ground set $\mV$ defined in (\ref{eq:ground_set}), and we have $\mC_r = \mC \cap \mV_r$. With this position, the cache capacity constraints in (\ref{eq:prob_b}) are equivalent to the condition $\mC \subseteq \m{I}$, where,
\begin{equation} \label{eq:matroid_const}
\m{I} = \left\{ {\mC \subseteq \mV:\left| {\mC \cap {\mV_r}} \right| \le {M_r},\forall r = 0,1,...R} \right\}.
\end{equation}

From (\ref{eq:matroid_def}) and (\ref{eq:matroid_const}) we can see that our constraints form a partition matroid $\m{M} = \left(\mV, \m{I} \right)$. In addition, notice from (\ref{eq:c_ir}) that the set $\left\{ {{c_{ir} (\mC)}:i \in \m{F}} \right\}$ can be considered as the Boolean representation of $\mC_r$. We now have the following Lemma:

\begin{lemma} The objective function in (\ref{eq:prob_a}) is a monotone submodular function.
\end{lemma}
\begin{proof}
For each file $f_i \in \m{F}$ and cache $\mC_r, r= 0,1,...R$, we introduce the new variables $t_{ir}^k$'s as: $t_{ir}^r = d_0$, $t_{ir}^0 = d_0 - d_r$, $t_{ir}^k = d_0 - d_{rk}$, $\forall k \in \mR\backslash r$. The objective function in (\ref{eq:prob_a}) can be expressed as,
\begin{equation}
\sum\limits_{r \in \mR} {\sum\limits_{u \in {\mU_r}} {\sum\limits_{i \in \m{F} } {{p_i}\tilde D_u^i} } },
\end{equation}
\begin{equation} \label{eq:Dt_ui}
\hspace*{-3.5cm} \text{where} \hspace{2cm} \tilde D_u^i = \sum\limits_{k = 0}^R {x_{ir}^kt_{ir}^k}.
\end{equation}

Since sum of monotone submodular functions is monotone submodular, it is enough to prove that for a user $u \in \mU_r$, the set function ${g_u}\left( \mC \right) = \tilde D_u^i$ is monotone submodular. Firstly, notice that from (\ref{eq:Dt_ui}), we have,
\begin{equation} \label{eq:tmax}
{g_u}\left( {\mC} \right) = \mathop {\max }\limits_{k \in \left\{ {0,1,...R} \right\}} t_{ir}^k\,\,\,{\rm{s}}{\rm{.t}}{\rm{.}}\,\,\,\,{c_{ik}}\left( \mC \right) = 1, \forall \mC \subseteq \mV.
\end{equation}

For a new file $f_{in} \in \mV \backslash \mC$, let ${\mC_{in}} = \mC + {f_{in}}$. It is straightforward to verify that ${g_u}\left( \mC_{in} \right) \geq {g_u}\left( \mC \right)$ and, therefore, ${g_u}\left( \mC \right)$ is a monotone function $\forall \mC \subseteq \mV$. Intuitively, adding a new file to a cache placement set cannot decrease the value of the set function. 

Let us now consider another cache placement set (decision) $\m{K}$ such that $\m{K} \subseteq \m{C}$. Denote ${\m{K}_{in}} = \m{K} + {f_{in}}$, we have ${g_u}\left( {\m{K}} \right) = t_{ir}^{\left( \m{K} \right)}$ and ${g_u}\left( \m{K}_{in} \right) = t_{ir}^{(\m{K}_{in})}$. Since ${g_u}\left( . \right)$ is monotone, we have,
\begin{equation} \label{eq:ine_monotone}
{g_u}\left( \mC \right) \ge {g_u}\left( \mK \right).
\end{equation}
The marginal value of adding the file $f_{in}$ to the sets $\mC$ and $\mK$ can be expressed, respectively, as,
\begin{equation}
{g_{u,\mC}}\left( {{f_{in}}} \right) = {g_u}\left( {\mC_{in} } \right) - {g_u}\left( \mC \right),
\end{equation}
\begin{equation}
{g_{u,\mK}}\left( {{f_{in}}} \right) = {g_u}\left( {\mK_{in} } \right) - {g_u}\left( \mK \right).
\end{equation}
In order to prove that ${g_u}\left( . \right)$ is submodular we need to show that ${g_{u,\mK}}\left( {{f_{in}}} \right) \ge {g_{u,\mC}}\left( {{f_{in}}} \right)$, or equivalently, that $\Delta _u^{in} = {g_{u,\mK}}\left( {{f_{in}}} \right) - {g_{u,\mC}}\left( {{f_{in}}} \right) \ge 0$. Using (\ref{eq:tmax}), we now distinguish three cases below,
\begin{itemize}
\item [(i)] $t_{ir}^n > {g_u}\left( \mC \right)$: We have ${g_u}\left( {{\mC_{in}}} \right) = {g_u}\left( {{\mK_{in}}} \right) = t_{ir}^n$. Thus, $\Delta _u^{in} = {g_u}\left( \mC \right) - {g_u}\left( \mK \right) \ge 0$, which stems from the inequality in (\ref{eq:ine_monotone}).

\item [(ii)] ${g_u}\left( \mK \right) \le t_{ir}^n \le {g_u}\left( \mC \right)$: In this case we have ${g_u}\left( {{\mC_{in}}} \right) = {g_u}\left( \mC \right)$ and ${g_u}\left( {{\mK_{in}}} \right) = t_{ir}^n$. Therefore, $\Delta _u^{in} = t_{ir}^n - {g_u}\left( \mK \right) \ge 0$.

\item [(iii)] $t_{ir}^n < {g_u}\left( \mK \right)$: In this case, adding $f_{in}$ does not provide any added value. We have ${g_u}\left( {{\mC_{in}}} \right) = {g_u}\left( \mC \right)$ and ${g_u}\left( {{\mK_{in}}} \right) = {g_u}\left( \mK \right)$. Thus $\Delta _u^{in} = 0$.
\end{itemize}

In summary, we always have $\Delta _u^{in} \geq 0$, which implies that ${g_u}\left( . \right)$ is submodular function in $\mV$. The proof is completed.
\end{proof}

A popular approach for the problem of maximizing a monotone submodular function subject to a matroid constraint is to use a greedy algorithm \cite{calinescu2011maximizing, nemhauser1978analysis}. Based on the result from Lemma~1, we can extent such algorithm to solve our problem in (\ref{eq:prob}). Our proposed solution consists of two phases: first, the content files are proactively distributed to the caches following the \emph{proactive cache distribution} algorithm; second, every time there is a cache miss and a new file is downloaded from the content server, the CCM will decide whether to replace this file with an existing ones in the caches following the \emph{reactive cache replacement} algorithm.

{\bf{Proactive cache distribution (PCD).}} The PCD algorithm incrementally builds a placement solution starting with the empty cache placement set. In each iteration it adds a new file with the highest marginal value to the cache placement set, until all the caches are full. Since the objective function is submodular, the marginal value of a new file decreases as the cache placement set grows bigger. We outline the procedure of the greedy PCD algorithm below.

\begin{algorithm} 
\caption{Proactive Cache Distribution.}\label{alg:pcd}
\renewcommand{\Statex}{\item[\hphantom{\bfseries Step \arabic{ALG@line}.}]}
\begin{algorithmic}[1]
\State Initialize: ${\mV_r} = \left\{ {{f_{1r}},{f_{2r}},...{f_{Fr}}} \right\}$, $\mC_r = \emptyset$, $r = 0,1,...R$, \par 
\hspace{0.8cm}$\mV = \left( {{\mV_0},{\mV_1},...{\mV_R}} \right)$, $\mC = \left( {{\mC_0},{\mC_1},...{\mC_R}} \right)$.

\Repeat
	\State ${f_{j'r'}} = \mathop {\arg \max }\limits_{{f_{jr}} \in \mV\backslash \mC} \left[ {g\left( {\mC + {f_{jr}}} \right) - g\left( \mC \right)} \right]$
	\State $\mC \leftarrow \mC + {f_{j'r'}}$
	\If{$\left| {{\mC_{r'}}} \right| = {M_{r'}}$} 
		$\mV \leftarrow \mV \backslash \mV_{r'}$
	\EndIf
\Until {$\mV = \emptyset$}
\State Output: $\mC$
\end{algorithmic}
\end{algorithm}

Step~3 of Algorithm~1 identifies the placement of file $j'$ in cache $\m{C}_{r'}$, denoted by $f_{j'r'}$, that provides the highest marginal value when adding to the current cache placement set $\m{C}$. Hence, $f_{j'r'}$ can be seen as the next best cache placement among the unplaced files $\left\{ {f_{jr} \in \m{V}\backslash \m{C}} \right\}$. In Algorithm~\ref{alg:pcd}, there would be $\sum\limits_{r = 0}^R {{M_r}}$ iterations until all the caches are full. Each iteration involves calculating the marginal value of at most $(R+1)F$ elements that have not been included in the cache set. Evaluating each marginal value would take $\m{O}(U)$ time. Hence, the running time would be $\m{O}\left( {\left( {R + 1} \right)FU\sum\limits_{r = 0}^R {{M_r}} } \right)$. When $M_r$ is a constant fraction of $F$, $\forall r =0,...R$, the time complexity is $\m{O}\left( {{{\left( {R + 1} \right)}^2}{F^2}U} \right)$. It has been shown that the greedy algorithm achieves a ratio of at least $\frac{1}{2}$ of the optimal value \cite{nemhauser1978analysis}.

{\bf{Reactive cache replacement (RCR).}} The PCD  algorithm described above initializes the cache distribution, which can be done during off-peak traffic hours (e.g., night-time) to utilize the unused backhaul bandwidth. Over the course of the day, following each cache miss, a new file will be downloaded from the remote content server to the BSs and delivered to the requesting user. As all the caches are full already, the CCM will decide to replace this new file with existing files in the caches only if such replacement could improve the value of the objective function. This approach ensures that the up-to-date cache placement set always yields the highest marginal value. The RCR algorithm is shown in Algorithm~\ref{alg:rcr}.
\begin{algorithm} 
\caption{Reactive Cache Replacement}\label{alg:rcr}
\renewcommand{\Statex}{\item[\hphantom{\bfseries Step \arabic{ALG@line}.}]}
\begin{algorithmic}[1]
\State For a new file request ${f_{i}} \notin \mC$
\For{$t = 0:R$}
	\State ${f_{j'r'}} = \mathop {\arg \min }\limits_{{f_{jr}} \in \mC} \left[ {g\left( \mC \right) - g\left( {\mC - {f_{jr}}} \right)} \right]$
	\If{$g\left( {\mC - {f_{j'r'}} + {f_{ir'}}} \right) > g\left( \mC \right)$}
		\State $\mC \leftarrow \mC - f_{j'r'} + f_{ir'}$
	\Else { Break}
	\EndIf
\EndFor
\State Output: $\mC$
\end{algorithmic}
\end{algorithm}

Step 3 of Algorithm~\ref{alg:rcr} evaluates the current utility value of every file in the cache placement set, which is computed by the utility value lost by removing that file from the current cache placement set. After that, Step 4 determines whether the new file has utility value greater than the least utility value of an existing file, and if so, replaces the new file with this least valued file. The running time of each iteration in Step 3 is $\m{O}\left( {U\sum\limits_{r = 0}^R {{M_r}} } \right)$. When $M_r$ is a constant fraction of $F$, the overall time complexity of the RCR algorithm is $\m{O}\left( {{{\left( {R + 1} \right)}^2}FU} \right)$

\begin{figure*}[t]
 \centering
 \begin{tabular}{ccc}
\hspace*{-.3cm}\includegraphics[scale = .6]{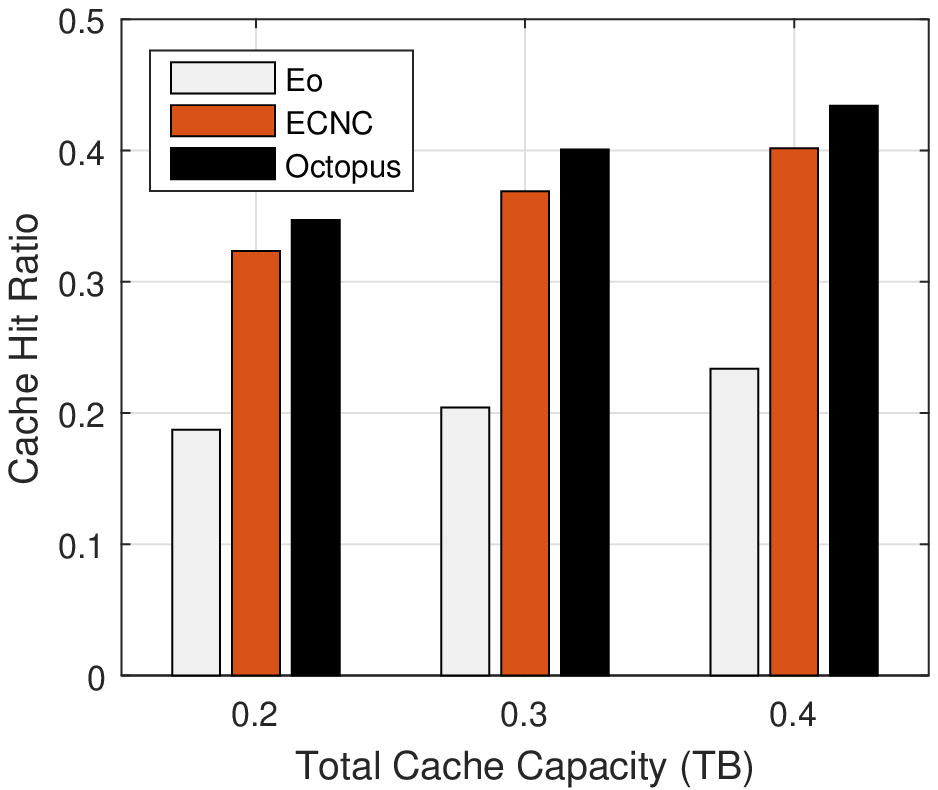} &
\hspace*{-.6cm}\includegraphics[scale = .6]{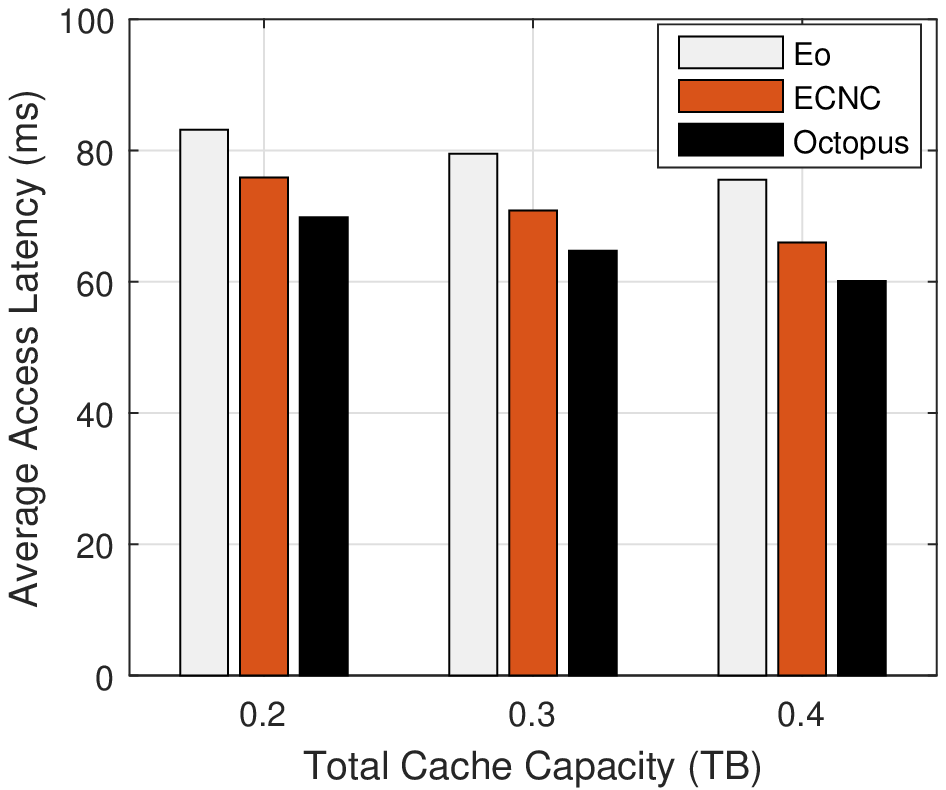} &
\hspace*{-.6cm}\includegraphics[scale = .6]{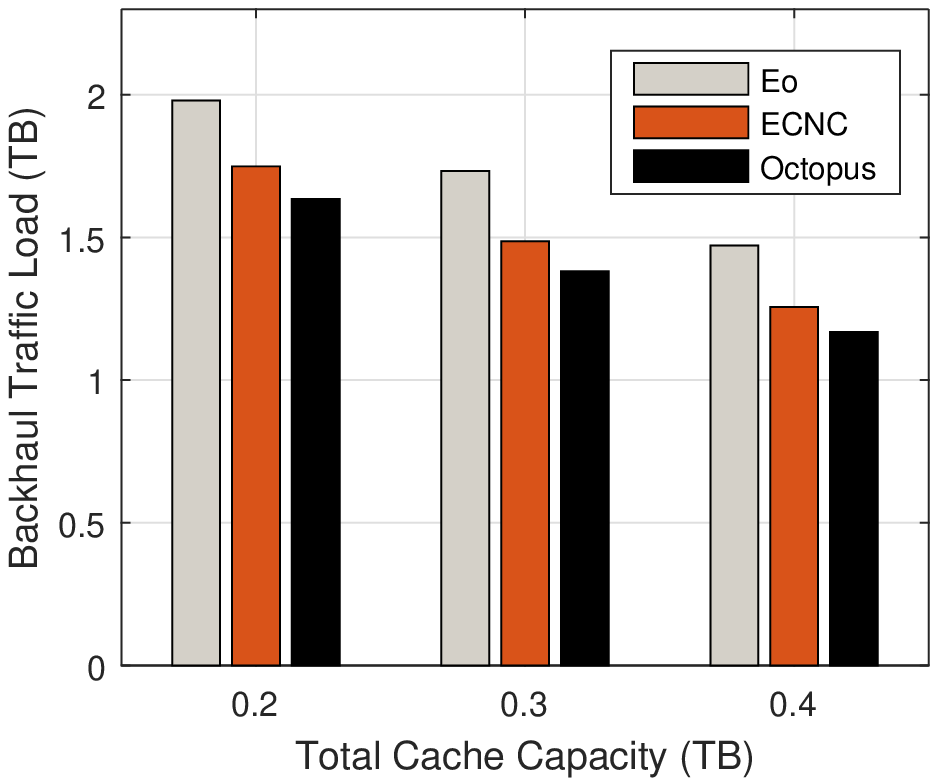} \\
 \small(a) & \small(b) & \small(c)
\end{tabular}
\caption{Performance comparison of different caching architectures: \emph{Eo} - the edge-only system with edge-caches at the BSs only; \emph{ECNC} - a non-cooperative hierarchical caching system with edge-caches and cloud-cache where each entity makes independent caching decision; \emph{Octopus} - the proposed cooperative hierarchical caching system.
}
\label{fig:cache_architecture}
\end{figure*}
\section{Performance Evaluation} \label{sec:perform_eval}

In this section, we present numerical simulations to evaluate the performance of the proposed Octopus caching system. We simulated a C-RAN system with $7$ cells, each having one BS, and mobile users are uniformly distributed in the cells. Unless otherwise stated, the simulation results are based on the YouTube request trace data collected on the University of Massachusetts' Amherst campus during the day $03/12/2008$~\cite{you2be_umass}. We consider the content files being the requested videos and the video popularities are extracted directly from the trace, which consists of $19,777$ users, $77,414$ videos and $122,280$ requests. It is assumed that the backhaul and fronthaul links' capacities as well as radio resources in the access network are sufficiently provisioned to handle all the generated traffic requests. The e2e latency of video delivery from the CDN to the CPU and from the CPU to BSs are randomly assigned, uniformly distributed in the ranges $\left[ {60 \div 100{\rm{ms}}} \right]$ and $\left[ {10 \div 30{\rm{ms}}} \right]$, respectively. In addition, we consider that the size of each video is $20~\rm{MB}$ and that $M_0 = 4M_r, \forall r = 1,...R$. 

%
%

We evaluate the considered caching schemes using three popular metrics: \emph{(i) cache hit ratio}: the fraction of requests that can be retrieved from one of the caches; \emph{(ii) average access delay $[\rm{ms}]$}: average latency of the contents traveling from the caches or the CDN server to the requesting user; \emph{(iii) backhaul traffic load $[\rm{TB}]$}: the volume of traffic going through the backhaul network due to users downloading contents from the CDN servers.

\subsection{Advantages of Cooperative Cloud-Cache}
In this subsection, we evaluate the benefit of using the cloud-cache cooperatively with the distributed edge-caches. Specifically, we compare the performance of the proposed Octopus caching system with the two traditional schemes below.


{\bf{Edge-only (Eo)}}: in this scheme, popular files are cached at the edge-caches only. If the requested file from a mobile user is found in the cache of its home BS, the file will be downloaded immediately from the cache; otherwise, it will be fetched from the CDN server.

{\bf{Edge+Cloud non-cooperative (ECNC)}}:  this is a hierarchical caching scheme where content placement decisions at the cloud-cache and edge-caches are made independently. A file request resulting in a cache miss in an edge-cache will be searched in the cloud-cache and finally goes out to the CDN. This scheme differs from Octopus in the sense that a user served by one BS cannot access the caches of other BSs. We employ the simplified version of PCD and RCR algorithms in Octopus to ECNC.

\begin{figure*}[t]
 \centering
 \begin{tabular}{ccc}
\hspace*{-.3cm}\includegraphics[scale = .6]{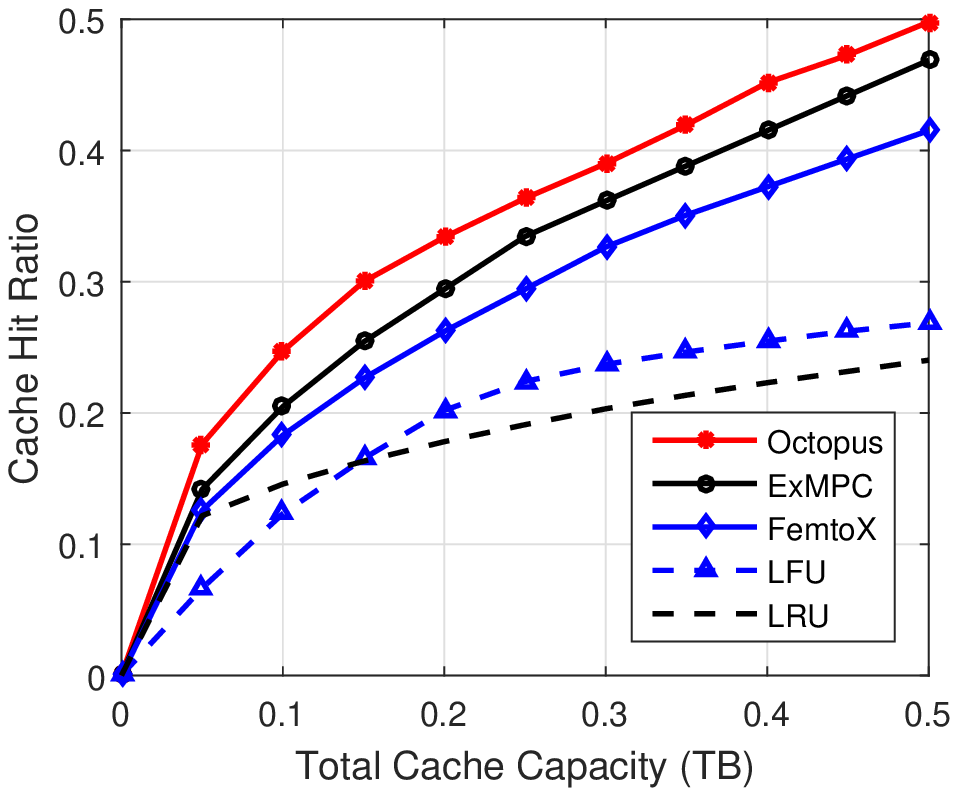} &
\hspace*{-.6cm}\includegraphics[scale = .6]{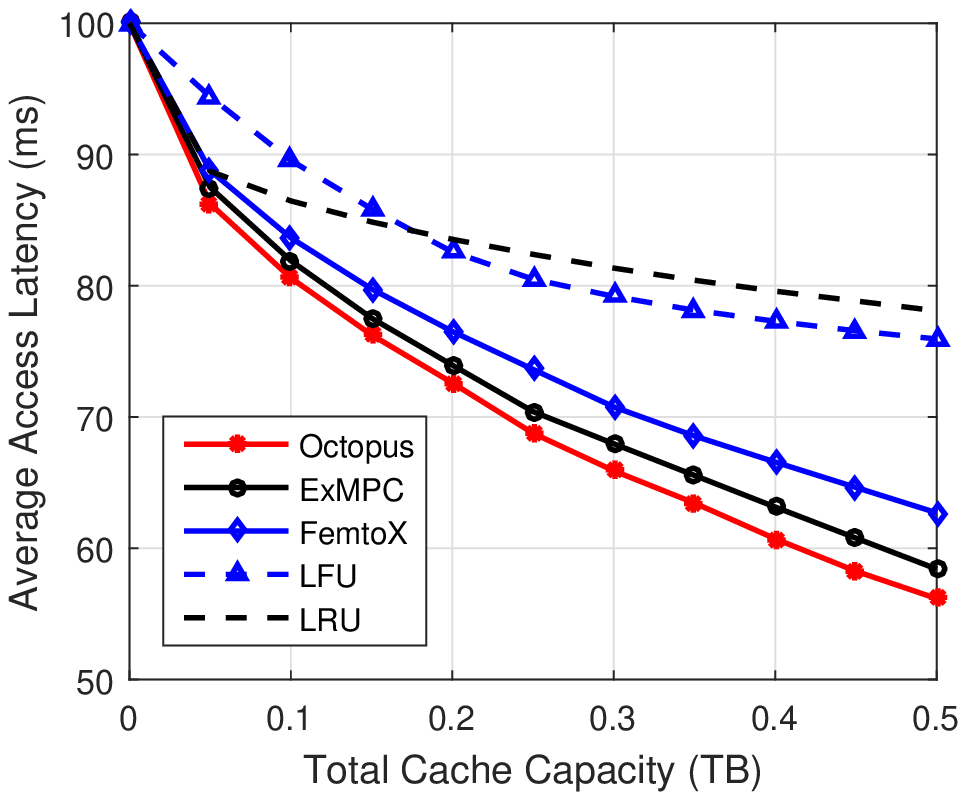} &
\hspace*{-.6cm}\includegraphics[scale = .6]{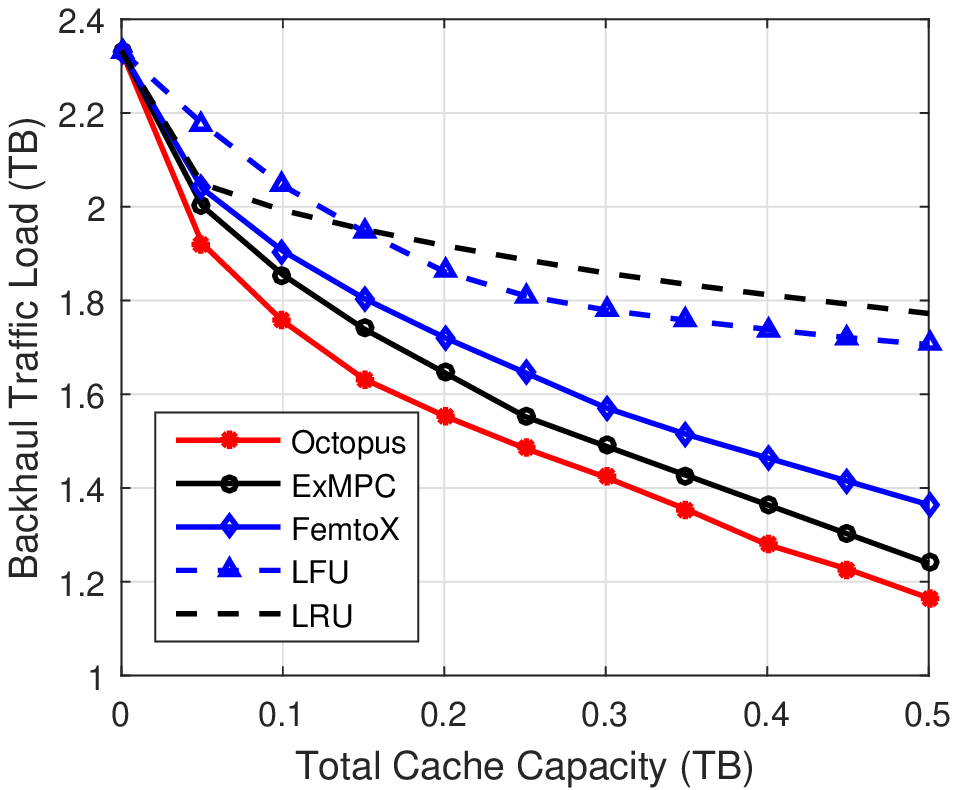} \\
 \small(a) & \small(b) & \small(c)
\end{tabular}
\caption{Performance of different cache management policies simulated using trace data.
}\label{fig:policy_trace}
\end{figure*}

Fig.~\ref{fig:cache_architecture}~(a-c) compares the performance of Octopus scheme with that of the two baseline schemes in terms of cache hit ratio, average access latency and backhaul traffic load, at different values of the total cache capacity. Observe that, owing to the additional cloud-cache layer, ECNC and Octopus schemes provide significant performance improvement compared to the Eo scheme with the same total cache capacity. For example, the performance gains of Octopus scheme over the Eo scheme when the total cache size is $0.4~\rm{TB}$ are approximately $80\%$ improvement in cache hit ratio, $21\%$ decrease in average access latency, and $20\%$ reduction in backhaul traffic load. In addition, the performance is further improved by cooperating the edge-caches, which are characterized by the gains of Octopus scheme over the ECNC scheme. 

\subsection{Advantages of Proposed Cache Management Algorithms}
In this subsection, we consider a system with both cloud-caches and edge-caches and evaluate different cache management algorithms, which identify the files to be stored in each cache. In particular, we compare the proposed Octopus scheme (which employs the PCD and RCR algorithms) with four baselines below.

\textbf{ExMPC}: This is the \emph{Exclusively Most Popular Caching} scheme where each edge-cache independently stores the most popular files, i.e., files with highest popularity. The cloud-cache stores the next most popular (second-tier) files that are not already stored in the edge-caches. This scheme aligns with the greedy algorithm presented in \cite{borst2010distributed} for inter-level cache cooperation. The exclusive mechanism in ExMPC avoids the redundancy in the pure MPC scheme \cite{cha2009analyzing} as the same files might be cached at both the edge and cloud layers.


\textbf{FemtoX}: This scheme is an extension of the FemtoCaching scheme~\cite{golrezaei2012femtocaching} to hierarchical caching system in cloud-based RAN. In FemtoCaching, the femtocell-like BSs act like the \emph{helpers} with weak backhaul links but large storage capacity. These helpers form a distributed caching network that assists the macro BS by handling requests and caches following a greedy algorithm. In FemtoX in this paper, we map each helper's cache in FemtoCaching to an edge-cache, and introduce the additional cloud-cache. 

\textbf{LFU}: This scheme is the application of the Least Frequently Used scheme \cite{lee2001lrfu} to hierarchical caching. 
When the cache is full and if there is a cache miss, LFU fetches the file from the CDN server and replaces it with the file in the cache that has been least frequently used.

\textbf{LRU}: This scheme is analogous to the LFU scheme; however when the cache is full, it chooses to evict the file that has been Least Recently Used. The cache hit ratio of LRU scheme depends on the overlap of content requests of the active users in the local RAN. 


Fig.~\ref{fig:policy_trace}~(a-c) compares the performance of Octopus caching scheme with the four baselines above. It can be seen that Octopus always achieves superior performance in terms of cache hit ratio, average access latency and backhaul traffic load. This is because the PCD algorithm helps further reduce the redundancy among the caches compared to ExMPC scheme, and due to the RCR algorithm helps updating the caches upon cache misses.

\begin{figure*}[t]
 \centering
 \begin{tabular}{ccc}
\hspace*{-.3cm}\includegraphics[scale = .6]{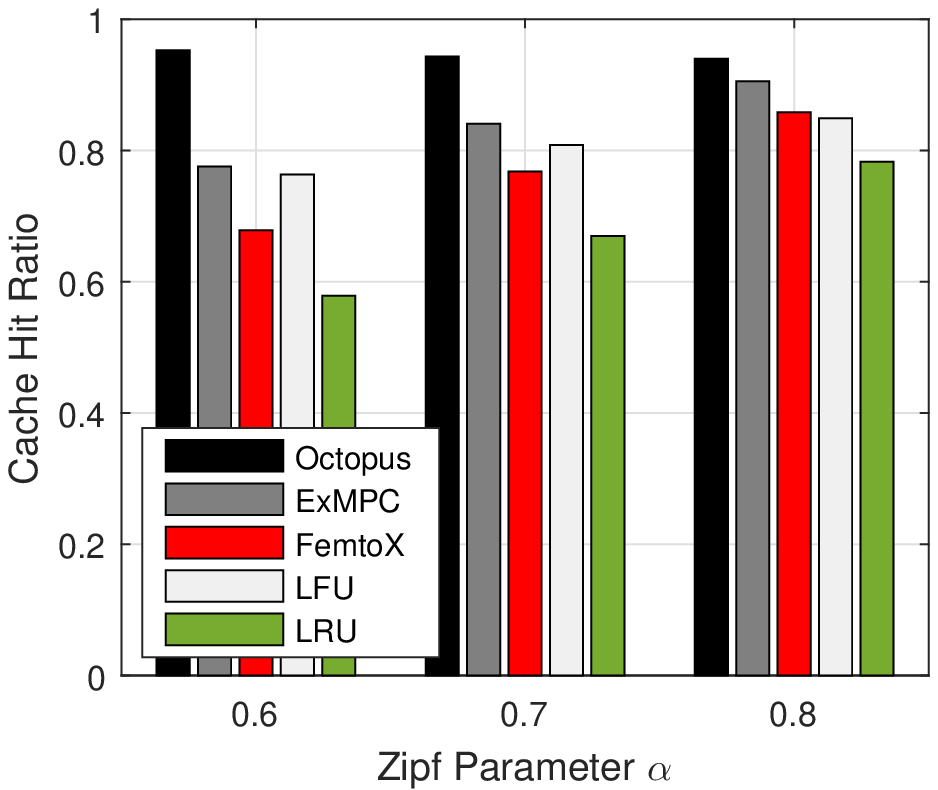} &
\hspace*{-.6cm}\includegraphics[scale = .6]{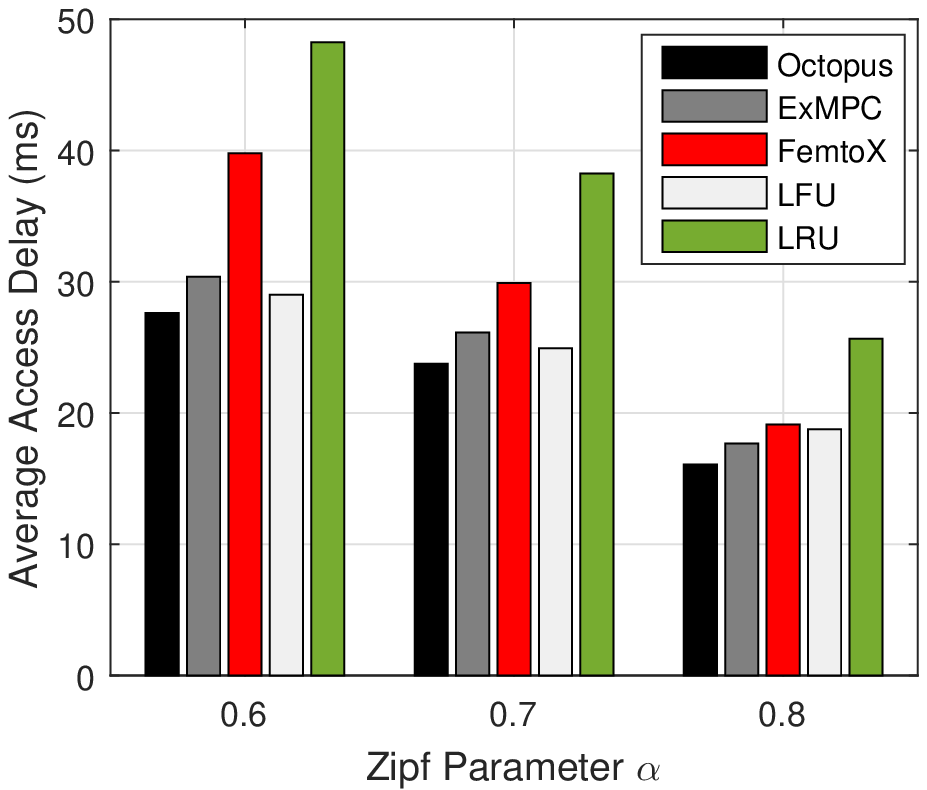} &
\hspace*{-.6cm}\includegraphics[scale = .6]{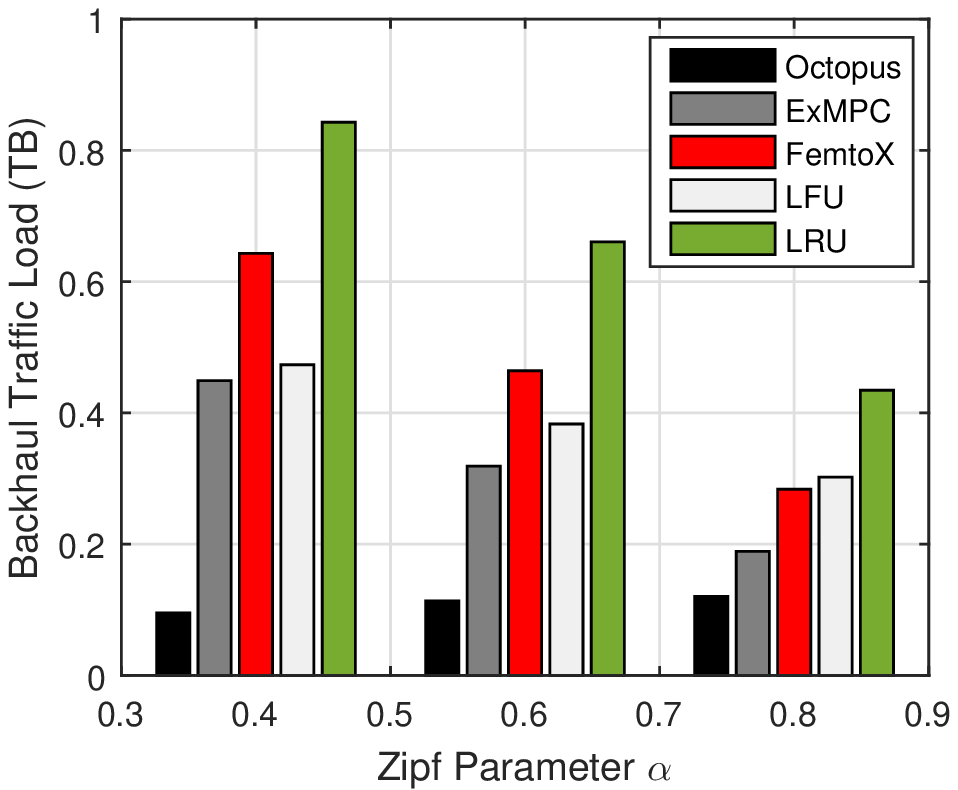} \\
 \small(a) & \small(b) & \small(c)
\end{tabular}
\caption{Performance of different cache management policies with synthetic file requests generated using the Zipf-based popularity distribution.
}\label{fig:policy_syn}
\end{figure*}

\subsection{Impact of Popularity Distributions}
In the previous subsections, using the YouTube request trace, we have demonstrated the superior performance of our proposed Octopus caching scheme over traditional caching architectures and management policies. In this subsection, to generalize the results, we evaluate the performance of Octopus using an analytical content request model. We consider that the popularity distribution of the content files follows a Zipf distribution. In particular, the request probability of the $k$th most popular content (among the set of $F$ contents) can be calculated as,
\begin{equation} \label{eq:Zipf}
{P_k} = \frac{{1/{k^\alpha }}}{{\sum\nolimits_{n = 1}^F {1/{n^\alpha }} }},
\end{equation}
where $\alpha$ is the Zipf parameter. The observed value of $\alpha$ might vary from different measurements, however it was estimated that $\alpha$ ranges from $0.64$ to $0.83$ based on the measurements of \cite{yu2006understanding, mahanti2000traffic}. To generate the synthetic requests, we consider a library of  $10,000$ files with uniform size equals to $20~\rm{MB}$. We randomly generate $100,000$ requests following the Zipf-based popularity distribution with $\alpha  \in \left[ {0.6,0.7,0.8} \right]$.

Fig.~\ref{fig:policy_syn}~(a-c) depicts the performance of Octopus caching scheme over the baselines with different content popularity distributions. Observe that as $\alpha$ increases, the performance of Octopus scheme in terms of cache hit ratio and backhaul traffic load is slightly degraded while its performance in terms of average access delay is significantly improved. In all cases, Octopus always performs the best. However, as $\alpha$ increases, the performance gaps between Octopus and the baselines become smaller.


%


\section{Conclusions} \label{sec:conclusion}
In this paper, we proposed Octopus, a cooperative hierarchical caching strategy for Cloud Radio Access Networks (C-RANs).  The proposed Octopus caching system consists of cloud-cache deployed at the CPU and distributed edge-caches deployed at the BSs. These caches are managed centrally by the CCM using the proactive cache distribution (PCD) and the reactive cache replaement (RCR) algorithms. We carried out extensive simulations using both the real world YouTube video request trace and the Zipf-based synthetic content request model. It is demonstrated that Octopus significantly outperforms traditional caching deployment architectures and cache management algorithms. Trace-driven simulations reveal that Octopus yields up to $80~\%$ improvement in cache hit ratio, $21~\%$ and $20~\%$ decrease in average access latency and backhaul traffic load, respectively, compared to the edge-only caching system using the same total cache capacity. 

\textbf{Acknowledgment: }This work was supported in part by the National Science Foundation Grant No. CNS-1319945

%
%
%

\section*{Appendix}
Here, we show that the cache placement optimization problem in (\ref{eq:prob}) is NP-complete. Firstly, notice that we have $FR\left( {R+2} \right) + (R + 1)$ constraints in (\ref{eq:prob}). We can easily verify the feasibility of any given solution in polynomial time by checking that the set of constraints is satisfied. Thus the problem is in NP. Following the approach in \cite{golrezaei2012femtocaching}, we will prove that the problem in (\ref{eq:prob}) is NP-hard by using a reduction from the set cover problem, which is known to be NP-complete problem, to an instance of our problem.

We consider a 2-disjoint set cover problem defined as follows. Given a bipartile graph $\m{G} = \left\{ {\mR,\mU,E} \right\}$ with edges $E$ connecting the set of two disjoint vertex sets $\mR$ and $\mU$ (i.e., $E$ only connects the element between two sets $\mR$ and $\mU$ while the elements in each set are disconnected). For each $r\in \m{R}$, define the subset $\m{N}_r \subseteq \m{U}$ containing the elements in $\m{U}$ that are connected to $r$ via $E$. Thus, clearly we have $\bigcup\limits_{r \in \m{R}} {{\m{N}_r}}  = {\m{U}}$. The objective of the 2-disjoint set cover problem is to find $\m{R}_1, \m{R}_2 \subseteq \m{R}$ such that ${{\m{R}}_1} \cap {\m{R}_2} = \emptyset$, $\left| {{\mR_1}} \right| + \left| {{\mR_2}} \right| = \left| \mR \right|$, and $\bigcup\nolimits_{r \in {\mR_1}} {{\mN_r}}  = \bigcup\nolimits_{r \in {\mR_2}} {{\mN_r}}  = \mU$. This problem is shown to be NP-complete \cite{cardei2005improving}.

The reduction from the 2-disjoint set cover problem to our cache placement problem in (\ref{eq:prob}) is done as follows: (i) the set of BSs is set to $\mR$, (ii) the set of users is set to $\mU$, (iii) the library of file is set to $\m{F} = \left\{ {{f_1},{f_2}} \right\}$ with the corresponding popularity $\mP = \left\{ {\frac{1}{{1 + \epsilon }},\frac{\epsilon }{{1 + \epsilon }}} \right\}$, (iv) the cache capacity at each BS is set to 1, (v) set the delay cost $d_0 = d_r = 1$, $d_{rk} = 0$ if $\mU_r \subseteq \mN_k$, and $d_{rk} = 1$ otherwise, $\forall r,k \in \m{R}$. We now show that there exists a solution to the 2-disjoint set cover problem if and only if there exists a solution to our problem in (\ref{eq:prob}) with objective value greater or equal to $U = \left| \mU \right|$. 

The first direction is easy to verify. Since all the caches have capacity $1$, they can either cache the file $f_1$ or $f_2$. If the two disjoint set covers $\mR_1, \mR_2$ exist, we can cache $f_1$ at all the BSs in $\mR_1$ and cache $f_2$ at all BSs in $\mR_2$. In this case, the objective value of our problem is equal to $U$. To prove the other direction, suppose we have a solution to our problem with objective value greater or equal to $U$, then it has to be equal to $U$ since the utility value seen by each user can be at most $1$. This is achieved only if the BSs caching $f_1$ and BSs caching $f_2$ cover the entire set of users $\mU$. This means there exist 2 disjoint set covers.

\bibliographystyle{ieeetr}
\bibliography{mobihoc16}

\end{document}